\documentclass[a4paper,final]{scrartcl}

\usepackage{lmodern}
\usepackage{amsmath}
\usepackage{amssymb}
\usepackage{amsthm}
\usepackage{hyperref}
\usepackage{xcolor}
\usepackage{enumitem}

\theoremstyle{plain}
\newtheorem{theorem}{Theorem}
\newtheorem{lemma}[theorem]{Lemma}
\newtheorem{corollary}[theorem]{Corollary}

\newtheorem{remark}[theorem]{Remark}

\renewcommand{\geq}{\geqslant}
\renewcommand{\leq}{\leqslant}
\renewcommand{\ge}{\geq}
\renewcommand{\le}{\leq}

\newcommand{\union}{\mathbin{\cup}}
\newcommand{\intersect}{\mathbin{\cap}}

\newcommand{\A}{\mathcal{A}}

\newcommand{\err}{\mathsf{err}}

\newcommand{\tr}{\mathrm{tr}}

\newcommand{\pad}{\diamond}

\newcommand{\N}{\mathbb{N}}

\newcommand{\ie}{i.e.,~}
\newcommand{\eg}{e.g.,~}
\newcommand{\ms}{\hspace*{0.5pt}}

\newcommand{\abs}[1] {\ensuremath\left|#1\right|}
\newcommand{\set}[2] {\ensuremath{\left\{#1 \mid #2\right\}}}
\newcommand{\os}[1] {\ensuremath{\left\{#1\right\}}}

\newcommand{\modd}[1]{\ (\mathrm{mod}\ #1)}

\begin{document}

\title{The State Complexity of Lexicographically Smallest Words and Computing Successors}

\author{Lukas Fleischer \and Jeffrey Shallit}

\date{School of Computer Science, University of Waterloo \\
200 University Avenue West, Waterloo, ON N2L 3G1, Canada \\
\texttt{\{lukas.fleischer,shallit\}@uwaterloo.ca}}

\maketitle

\begin{abstract}
  \noindent
  {\sffamily\normalsize\bfseries{Abstract.}} \
  Given a regular language $L$ over an ordered alphabet $\Sigma$, the set of lexicographically smallest (resp., largest) words of each length is itself regular.
  Moreover, there exists an unambiguous finite-state transducer that, on a given word $w \in \Sigma^*$, outputs the length-lexicographically smallest word larger than $w$ (henceforth called the \emph{$L$-successor} of $w$).
  In both cases, na\"ive constructions result in an exponential blowup in the number of states.
  We prove that if $L$ is recognized by a DFA with $n$ states, then $2^{\Theta(\sqrt{n \log n})}$ states are sufficient for a DFA to recognize the subset $S(L)$ of $L$ composed of its lexicographically smallest words. We give a matching lower bound that holds even if $S(L)$ is represented as an NFA.
  We then show that the same upper and lower bounds hold for an unambiguous finite-state transducer that computes $L$-successors.
\end{abstract}

\section{Introduction}

One of the most basic problems in formal language theory is the problem of enumerating the words of a language $L$.
Since, in general, $L$ is infinite, language enumeration is often formalized in one of the following two ways:
\begin{enumerate}
    \item A function that maps an integer $n \in \N$ to the $n$-th word of $L$.
    \item A function that takes a word and maps it to the next word in $L$.
\end{enumerate}
Both descriptions require some linear ordering of the words in order for them to be well-defined. Usually, \emph{radix order} (also known as length-lexicographical order) is used.
Throughout this work, we focus on the second formalization.

While enumeration is non-computable in general, there are many interesting special cases.
In this paper, we investigate the case of fixed regular languages, where
successors can be computed in linear time~\cite{AckermanM09,AckermanS09,Makinen97}.
Moreover, Frougny~\cite{Frougny97} showed that for every regular language~$L$, the mapping of words to their successors in~$L$ can be realized by a finite-state transducer.
Later, Angrand and Sakarovitch refined this result~\cite{AngrandS10}, showing that the successor function of any regular language is a finite union of functions computed by sequential transducers that operate from right to left.
However, to the best of our knowledge, no upper bound on the size of smallest transducer computing the successor function was known.

In this work, we consider transducers operating from left to right, and prove that the optimal upper bound for the size of transducers computing successors in $L$ is in $2^{\Theta(\sqrt{n \log n})}$, where $n$ is the size of the smallest DFA for~$L$.

The construction used to prove the upper bound relies heavily on another closely related result. Many years before Frougny published her proof, it had already been shown that if $L$ is a regular language, the set of all lexicographically smallest (resp., largest) words of each length is itself regular; see, \eg\cite{Sakarovitch83,Shallit94}.
This fact is used both in~\cite{AngrandS10} and in our construction.
In~\cite{Shallit94}, it was shown that if $L$ is recognized by a DFA with $n$ states, then the set of all lexicographically smallest words is recognized by a DFA with $2^{n^2}$ states.
While it is easy to improve this upper bound to $n \ms 2^n$, the exact state complexity of this operation remained open.
We prove that $2^{\Theta(\sqrt{n \log n})}$ states are sufficient and that this upper bound is optimal. We also prove that nondeterminism does not help with recognizing lexicographically smallest words, \ie the corresponding lower bound still holds if the constructed automaton is allowed to be nondeterministic.

The key component to our results is a careful investigation of the structure of lexicographically smallest words.
This is broken down into a series of technical lemmas in Section~\ref{sec:smallest}, which are interesting in their own right.
Some of the other techniques are similar to those already found in~\cite{AngrandS10}, but need to be carried out more carefully to achieve the desired upper bound.

For some related results, see \cite{Berthe:2007,Okhotin:2003}.

\section{Preliminaries}

We assume familiarity with basic concepts of formal language theory and automata theory; see~\cite{HopcroftMU06,Shallit08} for a comprehensive introduction. Below, we introduce concepts and notation specific to this work.

\paragraph{Ordered Words and Languages.}
Let $\Sigma$ be a finite ordered alphabet.
Throughout the paper, we consider words ordered by \emph{radix order}, which is defined by $u < v$ if either $\abs{u} < \abs{v}$ or there exist factorizations $u = xay$, $v = xbz$ with $\abs{y} = \abs{z}$ and $a, b \in \Sigma$ such that $a < b$.
We write $u \le v$ if $u = v$ or $u < v$. In this case, the word $u$ is \emph{smaller} than $v$ and the word $v$ is \emph{larger} than $u$.

For a language $L \subseteq \Sigma^*$ and two words $u, v \in \Sigma^*$, we say
that $v$ is the \emph{$L$-successor of $u$} if $v \in L$ and $w \not\in L$ for
all $w \in \Sigma^*$ with $u < w < v$.
Similarly, $u$ is the \emph{$L$-predecessor of $v$} if $u \in L$ and $w \not\in
L$ for all $w \in \Sigma^*$ with $u < w < v$.
A word is \emph{$L$-minimal} if it has no $L$-predecessor.
A word is \emph{$L$-maximal} if it has no $L$-successor.
Note that every nonempty language contains exactly one $L$-minimal word. It contains a (unique) $L$-maximal word if and only if $L$ is finite.
A word $u \in \Sigma^*$ is \emph{$L$-length-preserving} if it is not
$L$-maximal and the $L$-successor of $u$ has length $\abs{u}$. Words that are
not $L$-length-preserving are called \emph{$L$-length-increasing}.
Note that by definition, an $L$-maximal word is always $L$-length-increasing.
For convenience, we sometimes use the terms \emph{successor}
(resp.,~\emph{predecessor}) instead of \emph{$\Sigma^*$-successor}
(resp.,~\emph{$\Sigma^*$-predecessor}).

For a given language $L \subseteq \Sigma^*$, the set of all smallest words of each length in~$L$ is denoted by~$S(L)$. It is formally defined as follows:
\begin{equation*}
  S(L) = \set{u \in L}{\forall v \in L \colon v < u \implies \abs{v} < \abs{u}}.
\end{equation*}
Similarly, we define $B(L)$ to be the set of all $L$-length-increasing words:
\begin{equation*}
  B(L) = \set{u \in L}{\forall v \in L \colon v > u \implies \abs{v} > \abs{u}}.
\end{equation*}
A language $L \subseteq \Sigma^*$ is \emph{thin} if it contains at most one word of each length, \ie $\abs{L \intersect \Sigma^n} \in \os{0, 1}$ for all $n \ge 1$.
It is easy to see that for every language $L \subseteq \Sigma^*$, the languages $S(L)$ and $B(L)$ are thin.

\paragraph{Finite Automata and Transducers.}
A \emph{nondeterministic finite automaton} (NFA for short) is a $5$-tuple $(Q, \Sigma, {}\cdot{}, q_0, F)$ where $Q$ is a finite set of \emph{states}, $\Sigma$ is a finite alphabet, $q_0 \in Q$ is the \emph{initial state}, $F \subseteq Q$ is the set of \emph{accepting states} and ${}\cdot{} \colon Q \times \Sigma \to 2^Q$ is the \emph{transition function}. We usually use the notation $q \cdot a$ instead of ${}\cdot{}(q, a)$, and we extend the transition function to $2^Q \times \Sigma^*$ by letting $X \cdot \varepsilon = X$ and $X \cdot wa = \bigcup_{q \in X \cdot w}{q \cdot a}$ for all $X \subseteq Q$, $w \in \Sigma^*$, and $a \in \Sigma$.
For a state $q \in Q$ and a word $w \in \Sigma^*$, we also use the notation $q \cdot w$ instead of $\os{q} \cdot w$ for convenience.
A word $w \in \Sigma^*$ is \emph{accepted} by the NFA if $q_0 \cdot w \, \intersect\, F \ne \emptyset$.
We sometimes use the notation $p \xrightarrow{a} q$ to indicate that $q \in p \cdot a$.
An NFA is \emph{unambiguous} if for every input, there exists at most one accepting run. Unambiguous NFA are also called unambiguous finite state automata (UFA).
A \emph{deterministic finite automaton} (DFA for short) is an NFA $(Q, \Sigma, {}\cdot{}, q_0, F)$ with $\abs{q \cdot a} = 1$ for all $q \in Q$ and $a \in \Sigma$. Since this implies $\abs{q \cdot w} = 1$ for all $w \in \Sigma^*$, we sometimes identify the singleton $q \cdot w$ with the only element it contains.

A \emph{finite-state transducer} is a nondeterministic finite automaton that additionally produces some output that depends on the current state, the current letter and the successor state. For each transition, we allow both the input and the output letter to be empty.
Formally, it is a $6$-tuple $(Q, \Sigma, \Gamma, {}\cdot{}, q_0, F)$ where $Q$ is a finite set of states, $\Sigma$ and $\Gamma$ are finite alphabets, $q_0 \in Q$ is the initial state and $F \subseteq Q$ is the set of accepting states, and ${}\cdot{} \colon Q \times (\Sigma \union \os{\varepsilon}) \to 2^{Q \times (\Gamma \union \os{\varepsilon})}$ is the \emph{transition function}.
One can extend this transition function to the product $2^Q \times \Sigma^*$.
To this end, we first define the \emph{$\varepsilon$-closure} of a set $T \subseteq Q \times \Sigma^*$ as the smallest superset $C$ of $T$ with $\set{(q \cdot \varepsilon, w)}{(q, w) \in C} \subseteq C$.
We then define $X \cdot \varepsilon$ to be the $\varepsilon$-closure of $\set{(q, \varepsilon)}{q \in X}$ and $X \cdot wa$ to be the $\varepsilon$-closure of $\set{(q', ub)}{(q, u) \in X \cdot w, (q', b) \in q \cdot a}$ for all $X \subseteq Q$, $w \in \Sigma^*$ and $a \in \Sigma$.
We sometimes use the notation $p \xrightarrow{a \mid b} q$ to indicate that $(q, b) \in p \cdot a$.
A finite-state transducer is \emph{unambiguous} if, for every input, there exists at most one accepting run.

\section{The State Complexity of $S(L)$}
\label{sec:smallest}

It is known that if $L$ is a regular language, then both $S(L)$ and $B(L)$ are also regular~\cite{Sakarovitch83,Shallit94}.
In this section, we investigate the state complexity of the operations $L \mapsto S(L)$ and $L \mapsto B(L)$ for regular languages.
Since the operations are symmetric, we focus on the former.
To this end, we first prove some technical lemmas.
The first lemma is a simple observation that helps us investigate the structure of words in $S(L)$.

\begin{lemma}
  Let $x, u, y, v, z \in \Sigma^*$ with $\abs{u} = \abs{v}$. Then $xuuyz < xuyvz$ or $xyvvz < xuyvz$ or $xuuyz = xuyvz = xyvvz$.
  \label{lem:swap}
\end{lemma}
\begin{proof}
  Note that $uy$ and $yv$ are words of the same length.
  If $uy < yv$, then $xuuyz <xuyvz$. Similarly, $uy > yv$ immediately yields $xuyvz > xyvvz$. The last case is $uy = yv$, which implies $xuuyz = xuyvz = xyvvz$.
\end{proof}

Using this observation, we can generalize a well-known factorization technique for regular languages to minimal words.
For a DFA with state set $Q$, a state $q \in Q$ and a word $w = a_1 \cdots a_n \in \Sigma^*$, we define
\begin{align*}
    \tr(q, w) = (q, q \cdot a_1, \dots, q \cdot a_1 \cdots a_n)
\end{align*}
to be the sequence of all states that are visited when starting in state $q$ and following the transitions labeled by the letters from $w$.

\begin{lemma}
  Let $\A$ be a DFA over $\Sigma$ with $n$ states and with initial state $q_0$.
  Then for every word $w \in \Sigma^*$, there exists a factorization $w = u_1 v_1^{i_1} \cdots u_k v_k^{i_k}$ with $u_1, v_1, \dots, u_k, v_k \in \Sigma^*$ and $i_1, \dots, i_k \ge 1$ such that, for all $j \in \os{1, \dots, k}$, the following hold:
  \begin{enumerate}[label={(\alph*)}]
    \item\label{enum:loops} $q_0 \cdot u_1 v_1^{i_1} \cdots u_{j-1} v_{j-1}^{i_{j-1}} u_j = q_0 \cdot u_1 v_1^{i_1} \cdots u_{j-1} v_{j-1}^{i_{j-1}} u_j v_j$,
    \item\label{enum:len} $\abs{u_j v_j} \le n$, and
    \item\label{enum:prefix} $v_j$ is not a prefix of $u_{j+1} v_{j+1}^{i_{j+1}} \cdots u_k v_k^{i_k}$.
  \end{enumerate}
  Additionally, if $w \in S(L(\A))$, this factorization can be chosen such that
  \begin{enumerate}[label={(\alph*)},resume]
    \item\label{enum:disjoint-lengths} the lengths $\abs{v_j}$ are pairwise disjoint (\ie $\abs{\os{\abs{v_1}, \dots, \abs{v_k}}} = k$) and
    \item\label{enum:small-exponents} there exists at most one $j \in \os{1, \dots, k}$ with $i_j > n$.
  \end{enumerate}
  \label{lem:fact}
\end{lemma}
\begin{proof}
  To construct the desired factorization, initialize $j := 1$ and $q := q_0$ and follow these steps:
  \begin{enumerate}
    \item\label{enum:step1} If $w = \varepsilon$, we are done.
      If $w \ne \varepsilon$ and the states in $\tr(q, w)$ are pairwise distinct, let $u_j = w$ and $v_j = \varepsilon$ and we are done.
      Otherwise, factorize $w = xy$ with $\abs{x}$ minimal such that $\tr(q_0, x)$ contains exactly one state twice, \ie $\abs{x}$ distinct states in total.
    \item Choose the unique factorization $x = u v$ such that $q \cdot u = q \cdot uv$ and $v \ne \varepsilon$.
    \item Let $q := q \cdot x$ and $w := y$.
    \item If $j > 1$ and $u = \varepsilon$ and $v = v_{j-1}$, increment $i_{j-1}$ and go back to step~\ref{enum:step1}. Otherwise, let $u_j := u$, $v_j := v$ and $j := j + 1$; then go back to step~\ref{enum:step1}.
  \end{enumerate}
  This factorization satisfies the first three properties by construction.
  It remains to show that if $w \in S(L(\A))$, then Properties \ref{enum:disjoint-lengths} and \ref{enum:small-exponents} are satisfied as well.
  
  Let us begin with Property~\ref{enum:disjoint-lengths}.
  For the sake of contradiction, assume that there exist two indices $a, b$ with $a < b$ and $\abs{v_a} = \abs{v_b}$. Note that by construction, $v_a$ and $v_b$ must be nonempty.
  Moreover, by Property~\ref{enum:loops}, the words
  \begin{align*}
      w' & := u_1 v_1^{i_1} \cdots u_a v_a^{i_a+1} \cdots u_b v_b^{i_b-1} \cdots u_k v_k^{i_k} \text{~and} \\
      w'' & := u_1 v_1^{i_1} \cdots u_a v_a^{i_a-1} \cdots u_b v_b^{i_b+1} \cdots u_k v_k^{i_k}
  \end{align*}
  both belong to $L(\A)$.
  However, since $w \in S(L(\A))$, neither $w'$ nor $w''$ can be strictly smaller than $w$. Using Lemma~\ref{lem:swap}, we obtain that $w' = w$. This contradicts Property~\ref{enum:prefix}.
  
  Property~\ref{enum:small-exponents} can be proved by using the same argument: Assume that there exist indices $a, b$ with $a < b$ and $i_a, i_b > n$.
  The words $v_a^{\abs{v_b}}$ and $v_b^{\abs{v_a}}$ have the same lengths. We define
  \begin{align*}
      w' & := u_1 v_1^{i_1} \cdots u_a v_a^{i_a + \abs{v_b}} \cdots u_b v_b^{i_b- \abs{v_a}} \cdots u_k v_k^{i_k}, \\ 
      w'' & := u_1 v_1^{i_1} \cdots u_a v_a^{i_a - \abs{v_b}} \cdots u_b v_b^{i_b+\abs{v_a}} \cdots u_k v_k^{i_k},
  \end{align*}
  and obtain $w = w'$, which is a contradiction as above.
\end{proof}

The existence of such a factorization almost immediately yields our next technical ingredient.

\begin{lemma}
  Let $\A$ be a DFA with $n \ge 3$ states. Let $q_0$ be the initial state of $\A$ and let $w \in S(L(\A))$.
  Then there exists a factorization $w = xy^i z$ with $i \in \N$, $\abs{xz} \le n^3$ and $\abs{y} \le n$ such that $q_0 \cdot xy = q_0 \cdot x$.
  In particular, $xy^*z \subseteq L(\A)$.
  \label{lem:xyz}
\end{lemma}
\begin{proof}
  Let $w = u_1 v_1^{i_1} \cdots u_k v_k^{i_k}$ be a factorization that satisfies all properties in the statement of Lemma~\ref{lem:fact}.
  Suppose first that all exponents $i_j$ are at most~$n$.
  Using Properties~\ref{enum:len} and \ref{enum:disjoint-lengths}, we obtain $k \le n+1$ and the maximum length of~$w$ is achieved when all lengths $\ell \in \os{0, \dots, n}$ are present among the factors $v_j$ and the corresponding $u_j$ have lengths $n - \abs{v_j}$.
  This yields
  \begin{equation*}
    \abs{w} \le
    \sum_{\ell = 0}^n \big(n - \ell + n\ell\big) = 
    n(n + 1) + (n - 1) \sum_{\ell = 1}^n \ell = 
    n^2 + n + \frac{(n-1)n(n+1)}{2} \le n^3
  \end{equation*}
  where the last inequality uses $n \ge 3$.
  Therefore, we may set $x := w$, $y := \varepsilon$ and $z := \varepsilon$.
  
  If not all exponents are at most $n$, by Property~\ref{enum:small-exponents}, there exists a unique index $j$ with $i_j > n$. In this case, let $x := u_1 v_1^{i_1} \cdots u_{j-1} v_{j-1}^{i_{j-1}} u_j$, $y := v_j$ and $z := u_{j+1} v_{j+1}^{i_{j+1}} \cdots u_k v_k^{i_k}$.
  The upper bound $\abs{xz} \le n^3$ still follows by the argument above, and $\abs{y} \le n$ is a direct consequence of Property~\ref{enum:len}.
  Moreover, $w \in L(\A)$ and Property~\ref{enum:loops} together imply that $xy^*z \subseteq L(\A)$.
\end{proof}

For the next lemma, we need one more definition.
Let $\A$ be a DFA with initial state $q_0$.
Two tuples $(x, y, z)$ and $(x', y', z')$ are \emph{cycle-disjoint with respect to $\A$} if the sets of states in $\tr(q_0 \cdot x, y)$ and $\tr(q_0 \cdot x', y')$ are either equal or disjoint.

\begin{lemma}
  Let $\A$ be a DFA with $n \ge 3$ states and initial state $q_0$. 
  Let $(x, y, z)$ and $(x', y', z')$ be tuples that are not cycle-disjoint with respect to $\A$ such that
  \begin{equation*}
      q_0 \cdot x = q_0 \cdot xy, \quad 
      q_0 \cdot x' = q_0 \cdot x'y', \quad
      \abs{xz}, \abs{x'z'} \le n^3 \text{~and} \quad
      \abs{y}, \abs{y'} \le n.
  \end{equation*}
  Then either $xy^*z \intersect S(L(\A))$ or $x'(y')^*z' \intersect S(L(\A))$ only contains words of length at most $n^3 + n^2$.
  \label{lem:cycle-disjoint}
\end{lemma}
\begin{proof}
  Since the tuples are not cycle-disjoint with respect to $\A$, we can factorize $y = uv$ and $y' = u'v'$ such that $q_0 \cdot xu = q_0 \cdot x'u'$.

  Note that since $q_0 \cdot xuv = q_0 \cdot x$, the sets of states in $\tr(q_0 \cdot x, uv)$ and $\tr(q_0 \cdot xu, (vu)^i)$ coincide for all $i \ge 1$. By the same argument, the sets of states in $\tr(q_0 \cdot x', u'v')$ and $\tr(q_0 \cdot x'u', (v'u')^i)$ coincide for all $i \ge 1$.

  If the powers $(vu)^{\abs{y'}}$ and $(v'u')^{\abs{y}}$ were equal, then the sets of states in $\tr(q_0 \cdot xu, (vu)^{\abs{y'}})$ and $\tr(q_0 \cdot x'u', (v'u')^{\abs{y}})$ coincide. By the previous observation, this would imply that the tuples $(x, y, z)$ and $(x', y', z')$ are cycle-disjoint, a contradiction. We conclude $(vu)^{\abs{y'}} \ne (v'u')^{\abs{y}}$.
  
  By symmetry, we may assume that $(vu)^{\abs{y'}} < (v'u')^{\abs{y}}$.
  But then, for every word of the form $x' (y')^i z' \in L(\A)$ with $i > \abs{y}$, there exists a strictly smaller word
  $$x'u' (vu)^{\abs{y'}} (v'u')^{i-\abs{y}-1} v'z'$$
  in $L(\A)$.
  To see that this word indeed belongs to $L(\A)$, note that $q_0 \cdot x'u' vu = q_0 \cdot xuvu = q_0 \cdot xu = q_0 \cdot x'u'$.
  This means that all words in $x'(y')^*z' \intersect S(L)$ are of the form  $x' (y')^i z'$ with $i \le \abs{y}$.
\end{proof}

The previous lemmas now allow us to replace any language $L$ by another language that has a simple structure and approximates $L$ with respect to $S(L)$.

\begin{lemma}
  Let $\A$ be a DFA over $\Sigma$ with $n \ge 3$ states.
  Then there exist an integer $k \le n^4+n^3$ and tuples $(x_1, y_1, z_1), \dots, (x_k, y_k, z_k) \in (\Sigma^*)^3$ such that the following properties hold:
  \begin{enumerate}[label=(\roman*)]
      \item $S(L(\A)) \subseteq \bigcup_{i=1}^k x_i y_i^* z_i \subseteq L(\A)$,
      \item\label{enum:xzbound} $\abs{x_i z_i} \le n^3 + n^2$ for all $i \in \os{1, \dots k}$, and
      \item\label{enum:ybound} $\sum_{\ell \in Y} \ell \le n$ where $Y = \os{\abs{y_1}, \dots, \abs{y_k}}$.
  \end{enumerate}
  \label{lem:tl-simplify}
\end{lemma}
\begin{proof}
  If we ignore the required upper bound $k \le n^4+n^3$ and Property~\ref{enum:ybound} for now, the statement follows immediately from Lemma~\ref{lem:xyz} and the fact that there are only finitely many different tuples $(x, y, z)$ with $\abs{xz} \le n^3$ and $\abs{y} \le n$.
  We start with such a finite set of tuples $(x_1, y_1, z_1), \dots, (x_k, y_k, z_k)$ and show that we can repeatedly eliminate tuples until at most $n^4+n^3$ cycle-disjoint tuples remain.
  The desired upper bound $\sum_{\ell \in Y} \ell \le n$ then follows automatically.
  
  In each step of this elimination process, we handle one of the following cases:
  
  \begin{itemize}
    \item If there are two distinct tuples $(x_i, y_i, z_i)$ and $(x_j, y_j, z_j)$ with $\abs{x_i z_i} = \abs{x_j z_j}$ and $y_i = y_j$, there are two possible scenarios.
    If $x_i z_i < x_j z_j$, then for every word in $x_j y_j^* z_j$ there exists a smaller word in $x_i y_i^* z_i$ and we can remove $(x_j, y_j, z_j)$ from the set of tuples. By the same argument, we can remove the tuple $(x_i, y_i, z_i)$ if $y_i = y_j$ and $x_i z_i > x_j z_j$.
    
    \item Now consider the case that there are two distinct tuples $(x_i, y_i, z_i)$ and $(x_j, y_j, z_j)$ with $\abs{x_i z_i} = \abs{x_j z_j}$ and $\abs{y_i} = \abs{y_j}$ but $y_i \ne y_j$.
    We first check whether $x_i z_i < x_j z_j$. If true, we add the tuple $(x_i, \varepsilon, z_i)$, otherwise we add $(x_j, \varepsilon, z_j)$.
    If $x_i y_i < x_j y_j$, we know that each word in $x_j y_j^+ z_j$ has a smaller word in $x_i y_i^+ z_j$, and we remove the tuple $(x_j, y_j, z_j)$. Otherwise, we can remove $(x_i, y_i, z_i)$ by the same argument.
    
    \item
    The last case is that there exist two tuples $(x_i, y_i, z_i)$ and $(x_j, y_j, z_j)$ that are not cycle-disjoint. By Lemma~\ref{lem:cycle-disjoint}, we can remove at least one of these tuples and replace it by multiple tuples of the form $(x, \varepsilon, z)$.
    Note that the newly introduced tuples might be of the form $(x, \varepsilon, z)$ with $\abs{xz} > n^3$ but Lemma~\ref{lem:cycle-disjoint} asserts that they still satisfy $\abs{xz} \le n^3 + n^2$.
  \end{itemize}

  Note that we introduce new tuples of the form $(x, \varepsilon, z)$ during this elimination process. These new tuples are readily eliminated using the first rule.
  
  After iterating this elimination process, the remaining tuples are pairwise cycle-disjoint and the pairs $(\abs{x_iz_i}, \abs{y_i})$ assigned to these tuples $(x_i, y_i, z_i)$ are pairwise disjoint.
  Properties~\ref{enum:xzbound} and~\ref{enum:ybound} yield the desired upper bound on $k$.
\end{proof}

\begin{remark}
  While $S(L)$ can be approximated by a language of the simple form given in Lemma~\ref{lem:tl-simplify}, the language $S(L)$ itself does not necessarily have such a simple description.
  An example of a regular language $L$ where $S(L)$ does not have such a simple form is given in the proof of Theorem~\ref{thm:smallest-lower}.
\end{remark}

The last step is to investigate languages $L$ of the simple structure described in the previous lemma and show how to construct a small DFA for $S(L)$.

\begin{lemma}
  Let $n \in \N$.
  Let $L = \bigcup_{i=1}^k x_i y_i^* z_i$ with $k \le n^4+n^3$ and $\abs{x_i z_i} \le n^3 + n^2$ for all $i \in \os{1, \dots k}$ and $\sum_{\ell \in Y} \ell \le n$ where $Y = \os{\abs{y_1}, \dots, \abs{y_k}}$.
  Then $S(L)$ is recognized by a DFA with $2^{\mathcal{O}(\sqrt{n \log n})}$ states.
  \label{lem:tl-simplified-dfa}
\end{lemma}
\begin{proof}
  We describe how to construct a DFA of the desired size that recognizes the language $S(L)$. This DFA is the product automaton of multiple components.

  In one component (henceforth called the \emph{counter component}), we keep track of the length of the processed input as long as at most $n^3+n^2$ letters have been consumed.
  If more than $n^3+n^2$ letters have been consumed, we only keep track of the length of the processed input modulo all numbers $\abs{y_i}$ for $i \in \os{1, \dots, k}$.
  
  For each $i \in \os{1, \dots k}$, there is an additional component (henceforth called the \emph{$i$-th activity component}). In this component, we keep track of whether the currently processed prefix $u$ of the input is a prefix of a word in $x_i y_i^*$, whether $u$ is a prefix of a word in $x_i y_i^* z_i$ and whether $u \in x_i y_i^* z_i$.
  Note that if some prefix of the input is not a prefix of a word in $x_i y_i^* z_i$, no longer prefix of the input can be a prefix of a word in $x_i y_i^* z_i$.
  The information stored in the counter component suffices to compute the possible letters of $x_i y_i^* z_i$ allowed to be read in each step to maintain the prefix invariants.

  It remains to describe how to determine whether a state is final. To this end, we use the following procedure.
  First, we determine which sets of the form $x_i y_i^* z_i$ the input word leading to the considered state belongs to. These languages are called the \emph{active languages} of the state. They can be obtained from the activity components of the state. If there are no active languages, the state is immediately marked as not final.
  If the length of the input word $w$ leading to the considered state is $n^3 + n^2$ or less, we can obtain $\abs{w}$ from the counter component and reconstruct $w$ from the set of active languages.
  If the length of the input is larger than $n^3 + n^2$, we cannot fully recover the input from the information stored in the state. However, we can determine the shortest word $w$ with $\abs{w} > n^3 + n^2$ such that $\abs{w}$ is consistent with the length information stored in the counter component and $w$ itself is consistent with the set of active languages.
  In either case, we then compute the set $A$ of all words of length $\abs{w}$ that belong to any (possibly not active) language $x_i y_i^* z_i$ with $1 \le i \le k$.
  If $w$ is the smallest word in $A$, the state is final, otherwise it is not final.
  
  The desired upper bound on the number of states follows from known estimates on the least common multiple of a set of natural numbers with a given sum; see \eg \cite{Chrobak86}.
\end{proof}

We can now combine the previous lemmas to obtain an upper bound on the state complexity of $S(L)$.

\begin{theorem}
  Let $L$ be a regular language that is recognized by a DFA with $n$ states.
  Then $S(L)$ is recognized by a DFA with $2^{\mathcal{O}(\sqrt{n \log n})}$ states.
  \label{thm:smallest-upper}
\end{theorem}
\begin{proof}
  By Lemma~\ref{lem:tl-simplify}, we know that there exists a language $L'$ of the form described in the statement of Lemma~\ref{lem:tl-simplified-dfa} with $S(L) \subseteq L' \subseteq L$. 
  Since $L' \subseteq L$ implies $S(L') \subseteq S(L)$ and since $S(S(L)) = S(L)$, this also means that $S(L') = S(L)$. Lemma~\ref{lem:tl-simplified-dfa} now shows that there exists a DFA of the desired size.
\end{proof}

To show that the result is optimal, we provide a matching lower bound.

\begin{theorem}
  There exists a family of DFA $(\A_n)_{n \in \N}$ over a binary alphabet such that $\A_n$ has $n$ states and every NFA for $S(L(\A_n))$ has $2^{\Omega(\sqrt{n \log n})}$ states.
  \label{thm:smallest-lower}
\end{theorem}
\begin{proof}
  For $i \in \os{1, \dots k}$, let $p_i$ be the $i$-th prime number and let $p = p_1 \cdots p_k$. We define a language
  \begin{equation*}
    L = 1^* \union \bigcup_{1 \le i \le k} 1^i 0^{k-i+1} \os{1, 1^2, \dots, 1^{p_i-1}} (1^{p_i})^*.
  \end{equation*}
  It is easy to see that $L$ is recognized by a DFA with $k^2 + p_1 + \dots + p_k$ states.
  We show that $S(L)$ is not recognized by any NFA with less than $p$ states.
  From known estimates on the prime numbers (\eg \cite[Sec.~2.7]{Bach&Shallit:1996}), this suffices to prove our claim.
  
  Let $\A$ be a NFA for~$S(L)$ and assume, for the sake of contradiction, that~$\A$ has less than $p$ states.
  Note that since for each $i \in \os{1, \dots, k}$, the integer $p$ is a multiple of $p_i$, the language $L$ does not contain any word of the form $1^i 0^{k-i+1} 1^p$.
  Therefore, the word $1^{k+1+p}$ belongs to $S(L)$ and by assumption, an accepting path for this word in $\A$ must contain a loop of some length $\ell \in \os{1, \dots, p-1}$. But then $1^{k+1+p+\ell}$ is accepted by $\A$, too. However, since $1 \le \ell < p$, there exists some $i \in \os{1, \dots, k}$ such that $p_i$ does not divide $\ell$. This means that $p_i$ also does not divide $p + \ell$.
  Thus, $1^i 0^{k-i+1} 1^{p+\ell} \in L$, contradicting the fact that $1^{k + 1 + p + \ell}$ belongs to $S(L)$.
\end{proof}

Combining the previous two theorems, we obtain the following corollary.

\begin{corollary}
  Let $L$ be a language that is recognized by a DFA with $n$ states.
  Then, in general, $2^{\Theta(\sqrt{n \log n})}$ states are necessary and sufficient for a DFA or NFA to recognize $S(L)$.
\end{corollary}

By reversing the alphabet ordering, we immediately obtain similar results for largest words.

\begin{corollary}
  Let $L$ be a language that is recognized by a DFA with $n$ states.
  Then, in general, $2^{\Theta(\sqrt{n \log n})}$ states are necessary and sufficient for a DFA or NFA to recognize $B(L)$.
\end{corollary}

\section{The State Complexity of Computing Successors}

One approach to efficient enumeration of a regular language $L$ is
constructing a transducer that reads a word and outputs its $L$-successor~\cite{AngrandS10,Frougny97}.
We consider transducers that operate from left to right.
Since the output letter in each step might depend on letters that have not yet
been read, this transducer needs to be nondeterministic.
However, the construction can be made \emph{unambiguous}, meaning that for any given input, at most one computation path is accepting and yields the desired output word.
In this paper, we prove that, in general, $2^{\Theta(\sqrt{n \log n})}$ states are necessary and sufficient for a transducer that performs this computation.

Our proof is split into two parts.
First, we construct a transducer that only maps $L$-length-preserving words to their corresponding $L$-successors. All other words are rejected.
This construction heavily relies on results from the previous section.
Then we extend this transducer to $L$-length-increasing words by using a technique called \emph{padding}.
For the first part, we also need the following result.

\begin{theorem}
  Let $L \subseteq \Sigma^*$ be a thin language that is recognized by a DFA with~$n$ states. Then the languages
  \begin{align*}
    L_{\le} & = \set{v \in \Sigma^*}{\exists u \in L \colon \abs{u} = \abs{v} \text{and~} v \le u} \text{and} \\
    L_{\ge} & = \set{v \in \Sigma^*}{\exists u \in L \colon \abs{u} = \abs{v} \text{and~} v \ge u}
  \end{align*}
  are recognized by UFA with~$2n$ states.
  \label{thm:thin}
\end{theorem}
\begin{proof}
  Let $\A = (Q, \Sigma, {}\cdot{}, q_0, F)$ be a DFA for $L$ and let $n = \abs{Q}$.
  We construct a UFA with~$2n$ states for $L_{\le}$. The statement for $L_{\ge}$ follows by symmetry.
  
  The state set of the UFA is $Q \times \os{0, 1}$, the initial state is $(q_0, 0)$ and the set of final states is $F \times \os{0, 1}$.
  The transitions are
  \begin{alignat*}{2}
    (q, 0) &\xrightarrow{a} (q \cdot a, 0) && \qquad\text{for all $q \in Q$ and $a \in \Sigma$}, \\
    (q, 0) &\xrightarrow{a} (q \cdot b, 1) && \qquad\text{for all $q \in Q$ and $a, b \in \Sigma$ with $a < b$}, \\
    (q, 1) &\xrightarrow{a} (q \cdot b, 1) && \qquad\text{for all $q \in Q$ and $a, b \in \Sigma$}.
  \end{alignat*}
  
  It is easy to verify that this automaton indeed recognizes $L_{\le}$.
  To see that this automaton is unambiguous, consider an accepting run of a word $w$ of length~$\ell$.
  Note that the sequence of first components of the states in this run yield an accepting path of length $\ell$ in $\A$. Since $L(\A)$ is thin, this path is unique. Therefore, the sequence of first components is uniquely defined.
  The second components are then uniquely defined, too: they are $0$ up to the first position where~$w$ differs from the unique word of length $\ell$ in $L$, and $1$ afterwards.
\end{proof}

For a language $L \subseteq \Sigma^*$, we denote by $B_{\ge}(L)$ the language of all words from $\Sigma^*$ such that there exists no strictly larger word of the same length in $L$.
Combining Theorem~\ref{thm:smallest-upper} and Theorem~\ref{thm:thin}, the following corollary is immediate.

\begin{corollary}
  Let $L$ be a language that is recognized by a DFA with~$n$ states.
  Then there exists a UFA with $2^{\mathcal{O}(\sqrt{n \log n})}$ states that recognizes the language $B_{\ge}(L)$.
  \label{crl:upper-bigger}
\end{corollary}

For a language $L \subseteq \Sigma^*$, we define
\begin{equation*}
  X(L) = \set{u \in \Sigma^*}{\forall v \in L \colon \abs{u} \ne \abs{v}}.
\end{equation*}

If $L$ is regular, it is easy to construct an NFA for the complement of $X(L)$, henceforth denoted as $\overline{X(L)}$. To this end, we take a DFA for $L$ and replace the label of each transition with all letters from~$\Sigma$.
This NFA can also be viewed as an NFA over the unary alphabet $\os{\Sigma}$; here, $\Sigma$ is interpreted as a letter, not a set.
It can be converted to a DFA for $\overline{X(L)}$ by using Chrobak's efficient determinization procedure for unary NFA~\cite{Chrobak86}. The resulting DFA can then be complemented to obtain a DFA for $X(L)$:

\begin{corollary}
  Let $L$ be a language that is recognized by a DFA with~$n$ states.
  Then there exists a DFA with $2^{\mathcal{O}(\sqrt{n \log n})}$ states that recognizes the language $X(L)$.
  \label{crl:upper-exclude}
\end{corollary}

We now use the previous results to prove an upper bound on the size of a transducer performing a variant of the $L$-successor computation that only works for $L$-length-preserving words.

\begin{theorem}
  Let $L$ be a language that is recognized by a DFA with~$n$ states.
  Then there exists an unambiguous finite-state transducer with $2^{\mathcal{O}(\sqrt{n \log n})}$ states that rejects all $L$-length-increasing words and maps every $L$-length-preserving word to its $L$-successor.
  \label{thm:upper-length-preserving}
\end{theorem}
\begin{proof}
  Let $\A = (Q, \Sigma, {}\cdot{}, q_0, F)$ be a DFA for $L$ and let $n = \abs{Q}$.
  For every $q \in Q$, we denote by $\A_q$ the DFA that is obtained by making $q$ the new initial state of~$\A$.
  We use $\A^S_q$ to denote DFA with $2^{\mathcal{O}(\sqrt{n \log n})}$ states that recognizes the language $S(L(\A_q))$. These DFA exist by Theorem~\ref{thm:smallest-upper}.
  Moreover, by Corollary~\ref{crl:upper-bigger}, there exist UFA with $2^{\mathcal{O}(\sqrt{n \log n})}$ states that recognize the languages $B_{\ge}(L(\A_q))$. We denote these UFA by $\A^B_q$.
  Similarly, we use $\A^X_q$ to denote DFA with $2^{\mathcal{O}(\sqrt{n \log n})}$ states that recognize $X(L(\A_q))$. These DFA exist by Corollary~\ref{crl:upper-exclude}.
  
  In the finite-state transducer, we first simulate $\A$ on a prefix $u$ of the input, copying the input letters in each step, \ie producing the output $u$.
  At some position, after having read a prefix $u$ leading up to the state $q := q_0 \cdot u$, we nondeterministically decide to output a letter $b$ that is strictly larger than the current input letter $a$.
  From then on, we guess an output letter in each step and start simulating multiple automata in different components.
  In one component, we simulate $\A^B_{q \cdot a}$ on the remaining input.
  In another component, we simulate $\A^S_{q \cdot b}$ on the guessed output.
  In additional components, for each $c \in \Sigma$ with $a < c < b$, we simulate $\A^X_{q \cdot c}$ on the input.
  The automata in all components must accept in order for the transducer to accept the input.
  
  The automaton $\A^B_{q \cdot a}$ verifies that there is no word in $L$ that starts with the prefix $ua$, has the same length as the input word and is strictly larger than the input word.
  The automaton $\A^S_{q \cdot b}$ verifies that there is no word in $L$ that starts with the prefix $ub$, has the same length as the input word and is strictly smaller than the output word. It also certifies that the output word belongs to~$L$.
  For each letter $c$, the automaton $\A^X_{q \cdot c}$ verifies that there is no word in $L$ that starts with the prefix $uc$ and has the same length as the input word.

  Together, the components ensure that the guessed output is the unique successor of the input word, given that it is $L$-length-preserving. It is also clear that $L$-length-increasing words are rejected, since the $\A^S_{q \cdot b}$-component does not accept for any sequence of nondeterministic choices.
\end{proof}

The construction given in the previous proof can be extended to also compute $L$-successors of $L$-length-increasing words. However, this requires some quite technical adjustments to the transducer. Instead, we use a technique called \emph{padding}. A very similar approach appears in~\cite[Prop.~5.1]{AngrandS10}.

We call the smallest letter of an ordered alphabet $\Sigma$ the \emph{padding symbol} of~$\Sigma$.
A language $L \subseteq \Sigma^*$ is \emph{$\pad$-padded} if $\pad$
is the padding symbol of $\Sigma$ and $L = \pad^* K$ for some $K \subseteq (\Sigma
\setminus \os{\pad})^*$.
The key property of padded languages is that all words prefixed by a sufficiently long
block of padding symbols are $L$-length-preserving.

\begin{lemma}
  Let $\A$ be a DFA over $\Sigma$ with $n$ states
  such that $L(\A)$ is a $\pad$-padded language. Let $\Gamma = \Sigma \setminus
  \os{\pad}$ and let $K = L(\A) \intersect \Gamma^*$.
  Let $u \in \Gamma^*$ be a word that is not $K$-maximal.
  Then the $L(\A)$-successor of $\pad^n u$ has length $\abs{\pad^n u}$.
  \label{lem:same-length}
\end{lemma}
\begin{proof}
  Let $v$ be the $K$-successor of $u$. By a standard pumping argument, we have
  $\abs{u} \le \abs{v} \le \abs{u} + n$. This means that $\pad^{n + \abs{u} -
  \abs{v}} v$ is well-defined and belongs to $L(\A)$. Note that this word is
  strictly greater than $\pad^n u$ and has length $\abs{\pad^n u}$.
  Thus, the $L(\A)$-successor of $\pad^n u$ has length $\abs{\pad^n u}$, too.
\end{proof}

We now state the main result of this section.

\begin{theorem}
  Let $\A$ be a deterministic finite automaton over $\Sigma$ with $n$ states.
  Then there exists an unambiguous finite-state transducer with $2^{\mathcal{O}(\sqrt{n \log
  n})}$ states that maps every word to its $L(\A)$-successor.
\end{theorem}
\begin{proof}
  We extend the alphabet by adding a new padding symbol $\diamond$ and convert $\A$ to a DFA for $\diamond^* L$ by adding a new
  initial state.
  The language $L'$ accepted by this new DFA is $\pad$-padded. By
  Theorem~\ref{thm:upper-length-preserving} and Lemma~\ref{lem:same-length}, there exists an unambiguous transducer of the
  desired size that maps every word from $\pad^{n+1} \Sigma^*$ to its successor in $L'$.
  It is easy to modify this transducer such that all words that do not belong to $\pad^{n+1} \Sigma^*$ are rejected.
  We then replace every transition that reads a $\pad$ by a corresponding
  transition that reads the empty word instead.
  Similarly, we replace every transition that outputs a $\pad$ by a transition
  that outputs the empty word instead.
  Clearly, this yields the desired construction for the original language $L$.
  A careful analysis of the construction shows that the transducer remains unambiguous after each step.
\end{proof}

We now show that this construction is optimal up to constants in the exponent.
The idea is similar to the construction used in Theorem~\ref{thm:smallest-lower}.
\begin{theorem}
  There exists a family of deterministic finite automata $(\A_n)_{n \in \N}$
  such that $\A_n$ has $n$ states whereas the smallest unambiguous transducer
  that maps every word to its $L(\A_n)$-successor has $2^{\Omega(\sqrt{n \log
  n})}$ states.
\end{theorem}
\begin{proof}
  Let $k \in \N$.
  Let $p_1, \dots, p_k$ be the $k$ smallest prime numbers such that $p_1 <
  \cdots < p_k$ and let $p = p_1 \cdots p_k$.
  We construct a deterministic finite automaton $\A$ with $2 + p_1 + \dots
  + p_k$ states such that the smallest transducer computing the desired mapping
  has at least $p$ states. From known estimates on the prime numbers (\eg \cite[Sec.~2.7]{Bach&Shallit:1996}), this suffices to prove our claim.

  The automaton is defined over the alphabet $\Sigma = \os{1, \dots, k} \union
  \os{\#}$.
  It consists of an initial state $q_0$, an error state $q_\err$, and
  states $(i, j)$ for $i \in \os{1, \dots, k}$ and $j \in \os{0, \dots,
  p_i-1}$ with transitions defined as follows:
  \begin{align*}
      q_0 \cdot a & = 
        \begin{cases}
          (a, 0), & \text{for $a \in \os{1, \dots, k}$}; \\
          q_\err, & \text{if $a = \#$}; \\
        \end{cases} \\
      (i, j) \cdot a & = 
        \begin{cases}
          (i, j+1 \bmod p_i), & \text{if $a = \#$}; \\
          q_\err, & \text{for $a \in \os{1, \dots, k}$}.
        \end{cases}
  \end{align*}
  The set of accepting states is $\set{(i, 0)}{1 \le i \le k}$.
  The language $L(\A)$ is the set of all words of the form $i \#^{j}$ with $1
  \le i \le k$ such that $j$ is a multiple of $p_i$.

  Assume, to get a contradiction, that there exists an unambiguous
  transducer with less than $p$ states that maps $w$ to the smallest word in
  $L(\A)$ strictly greater than $w$.
  Consider an accepting run of this transducer on some input of the form $2
  \#^{\ell p}$ with $\ell \in \N$ large enough such that the run contains a
  cycle.
  Clearly, since $\ell p + 1$ and $p$ are coprime, the output of the transducer
  has to be $2\#^{\ell p + 2}$.
  We fix one cycle in this run.

  If the number of $\#$ read in this cycle does not equal the number of $\#$ output in
  this cycle, by using a pumping argument, we can construct a word of the form
  $2 \#^j$ that is mapped to a word or the form $i \#^{j'}$ with $\abs{j'-j} >
  2$. This contradicts the fact that $2 \#^{2\N}$ is a subset of $L(\A)$.
  Therefore, we may assume that both the number of letters read and output on
  the cycle is $r \in \os{1, \dots, p-1}$.

  Again, by a pumping argument, this implies that $2\#^{\ell p + jr}$ is mapped
  to $2\#^{\ell p + jr + 2}$ for every $j \in \N$.
  Since $r < p$, at least one of the prime numbers $p_i$ is coprime to $r$.
  Therefore, we can choose $j$ such that $jr + 1 \equiv 0\modd {p_i}$.
  However, this means that $p_i \#^{\ell p + jr + 1}$ belongs to $L(\A)$,
  contradicting the fact that the transducer maps $2\#^{\ell p + jr}$ to
  $2\#^{\ell p + jr + 2}$.
\end{proof}

Combining the two previous theorems, we obtain the following corollary.

\begin{corollary}
  Let $L$ be a language that is recognized by a DFA with $n$ states.
  Then, in general, $2^{\Theta(\sqrt{n \log n})}$ states are necessary and sufficient for an unambiguous finite-state transducer that maps words to their $L$-successors.
\end{corollary}

\end{document}